\documentclass[11pt]{article} \usepackage[normal]{optional}
\usepackage{amsmath}
\usepackage{amssymb}
\usepackage{ifpdf}
\usepackage{algorithmic}
\usepackage[caption=false,format=hang]{subfig}
\usepackage{wrapfig}
\usepackage{cite}
\usepackage{graphicx}
\usepackage{enumitem}
\usepackage{comment}
\usepackage{color}

\opt{normal}{\usepackage{commands-tam,numbering-tam}}
\opt{submission}{
    \usepackage{commands-tam}
    \numberwithin{equation}{section}
    \numberwithin{theorem}{section}
    \spnewtheorem{cor}[theorem]{Corollary}{\bfseries}{\itshape}
    \spnewtheorem{lem}[theorem]{Lemma}{\bfseries}{\itshape}
    \spnewtheorem{prop}[theorem]{Proposition}{\bfseries}{\itshape}
    \spnewtheorem{obs}[theorem]{Observation}{\bfseries}{\itshape}
    \renewcommand{\vec}[1]{}
    \usepackage{pageno}
}

\newcommand{\qedl}{\opt{normal}{}\opt{submission,final}{\qed}}

\usepackage{url}
\usepackage{amsfonts}
\opt{normal}{\usepackage{amsthm}}
\opt{normal}{\usepackage{fullpage}}

\opt{submission,final}{\newcommand{\figuresize}{\footnotesize}}
\opt{normal}{\newcommand{\figuresize}{\small}}

\newcommand{\dN}{\mathsf{N}}
\newcommand{\dS}{\mathsf{S}}
\newcommand{\dE}{\mathsf{E}}
\newcommand{\dW}{\mathsf{W}}
\newcommand{\dall}{\{\dN,\dS,\dE,\dW\}}

\newcommand{\ignore}[1]{}

\newcommand{\ha}{\widehat{\alpha}}

\ifpdf

  \usepackage[pdftex]{epsfig}
  \usepackage[pdfstartview=FitV,pdfpagemode=None,colorlinks=true,citecolor=blue,linkcolor=blue]{hyperref}

\else
    \usepackage[dvips]{epsfig}
\usepackage[dvips,pdfstartview=FitH,pdfpagemode=None,colorlinks=true,citecolor=blue,linkcolor=blue]{hyperref}

\fi

\begin{document}

\opt{normal}{
    \title{Program Size and Temperature in Self-Assembly\thanks{The first author was supported by the Molecular Programming Project under NSF grant 0832824, the second author was supported by an NSF Computing Innovation Fellowship, and the third author was supported by NSERC Discovery Grant R2824A01 and the Canada Research Chair in Biocomputing to Lila Kari.}}
    \date{}

    \author{
        Ho-Lin Chen\thanks{California Institute of Technology, Department of Computing and Mathematical Sciences, Pasadena, CA, USA, {\tt holinc@gmail.com}, {\tt ddoty@caltech.edu}}
        \and
        David Doty\footnotemark[2]
        \and
        Shinnosuke Seki\thanks{University of Western Ontario, Department of Computer Science, London, ON, Canada, N6A 5B7,  {\tt sseki@csd.uwo.ca}.}
    }
}

\opt{submission}{
    \title{Program Size and Temperature in Self-Assembly\thanks{The first author was supported by the Molecular Programming Project under NSF grant 0832824, the second author was supported by an NSF Computing Innovation Fellowship, and the third author was supported by NSERC Discovery Grant R2824A01 and the Canada Research Chair in Biocomputing to Lila Kari.}}
    \date{}

    \author{
        Ho-Lin Chen\inst{1}
        \and
        David Doty\inst{1}
        \and
        Shinnosuke Seki\inst{2}
    }

    \institute{California Institute of Technology, Department of Computing and Mathematical Sciences, Pasadena, CA, USA, \email{holinc@gmail.com}, \email{ddoty@caltech.edu},
    \and
    University of Western Ontario, Department of Computer Science, London, ON, Canada, \email{sseki@csd.uwo.ca}}
}

\maketitle


\begin{abstract}

Winfree's abstract Tile Assembly Model (aTAM) is a model of molecular self-assembly of DNA complexes known as tiles, which float freely in solution and attach one at a time to a growing ``seed'' assembly based on specific binding sites on their four sides.
We show that there is a polynomial-time algorithm that, given an $n \times n$ square, finds the minimal tile system (i.e., the system with the smallest number of distinct tile types) that uniquely self-assembles the square, answering an open question of Adleman, Cheng, Goel, Huang, Kempe, Moisset de Espan\'{e}s, and Rothemund (\emph{Combinatorial Optimization Problems in Self-Assembly}, STOC 2002).
Our investigation leading to this algorithm reveals other positive and negative results about the relationship between the size of a tile system and its ``temperature'' (the binding strength threshold required for a tile to attach).
\end{abstract}







\opt{submission,final}{\thispagestyle{empty}\newpage\setcounter{page}{1} \clearpage}
\section{Introduction}
\label{sec-intro}

Tile self-assembly is an algorithmically rich model of ``programmable crystal growth''.
It is possible to design monomers (square-like ``tiles'') with specific binding sites so that, even subject to the chaotic nature of molecules floating randomly in a well-mixed chemical soup, they are guaranteed to bind so as to deterministically form a single target shape.
This is despite the number of different types of tiles possibly being much smaller than the size of the shape and therefore having only ``local information'' to guide their attachment.
The ability to control nanoscale structures and machines to atomic-level precision will rely crucially on sophisticated self-assembling systems that automatically control their own behavior where no top-down externally controlled device could fit.

A practical implementation of self-assembling molecular tiles was proved experimentally feasible in 1982 by Seeman~\cite{Seem82} using DNA complexes formed from artificially synthesized strands.
Experimental advances have delivered increasingly reliable assembly of algorithmic DNA tiles with error rates of 10\% per tile in 2004~\cite{RoPaWi04}, 1.4\% in 2007~\cite{FujHarParWinMur07}, 0.13\% in 2009~\cite{BarSchRotWin09}, and 0.05\% in 2010~\cite{Constantine10personal}.
Erik Winfree~\cite{Winf98} introduced the abstract Tile Assembly Model (aTAM) -- based on a constructive version of Wang tiling~\cite{Wang61,Wang63} -- as a simplified mathematical model of self-assembling DNA tiles.
Winfree demonstrated the computational universality of the aTAM by showing how to simulate an arbitrary cellular automaton with a tile assembly system.
Building on these connections to computability, Rothemund and Winfree~\cite{RotWin00} investigated a self-assembly resource bound known as \emph{tile complexity}, the minimum number of tile types needed to assemble a shape.
They showed that for most $n$, the problem of assembling an $n \times n$ square has tile complexity $\Omega(\frac{\log n}{\log \log n})$, and Adleman, Cheng, Goel, and Huang~\cite{AdChGoHu01} exhibited a construction showing that this lower bound is asymptotically tight.
Under natural generalizations of the model~\cite{AGKS05, BeckerRR06, KaoSchS08, KS06, DDFIRSS07, Sum09, Dot10, ChaGopRei09, RNaseSODA2010, DemPatSchSum2010RNase, SolWin07, ManStaSto09ISAAC}, tile complexity can be reduced for the assembly of squares and more general shapes.

The results of this paper are motivated by the following problem posed in 2002.
Adleman, Cheng, Goel, Huang, Kempe, Moisset de Espan\'{e}s, and Rothemund~\cite{ACGHKMR02} showed that there is a polynomial time algorithm for finding a minimum size tile system (i.e., system with the smallest number of distinct tile types) to uniquely self-assemble a given $n \times n$ square,\footnote{The square is encoded by a list of its points, so the algorithm's running time is polynomial in $n$.} subject to the constraint that the tile system's ``temperature'' (binding strength threshold required for a tile to attach) is 2.
They asked whether their algorithm could be modified to remove the temperature 2 constraint.
Our main theorem answers this question affirmatively.
Their algorithm works by brute-force search over the set of all temperature-2 tile systems with at most $O(\frac{\log n}{\log \log n})$ tile types, using the fact proven by Adleman, Cheng, Goel, and Huang~\cite{AdChGoHu01} that such an upper bound on tile complexity suffices to assemble any $n \times n$ square.
A simple counting argument shows that for any constant $\tau$, the number of tile systems with glue strengths and temperature at most $\tau$ and $O(\frac{\log n}{\log \log n})$ tile types is bounded by a polynomial in $n$.
One conceivable approach to extending the algorithm to arbitrary temperature is to prove that for any tile system with $K$ tile types, the strengths and temperature can be re-assigned so that they are upper-bounded by a constant or slow-growing function of $K$, without affecting the behavior of the tile system.\footnote{We define ``behavior'' more formally in
Section~\ref{sec-find-strengths}.
Briefly, we consider a tile system's behavior unaltered by a reassignment of strengths and temperature if, for each tile type $t$, the reassignment has not altered the collection of subsets of sides of $t$ that have sufficient strength to bind.}
However, we show that this approach cannot work, by demonstrating that for each $K$, there is a tile system with $K$ tile types whose behavior cannot be preserved using any temperature less than $2^{K/4}$.
The proof crucially uses \emph{3-cooperative} binding, meaning attachment events that require three different glues of a tile to match the assembly.
On the other hand, we show that any 2-cooperative tile system (which captures a wide and useful class of systems) with $K$ tile types is behaviorally equivalent to a system with temperature at most $2K+2$.

Of course, the choice of integer glue strengths is an artifact of the model.
Nonetheless, our investigation does reflect fundamental questions about how finely divided molecular binding energies must be in a real molecular self-assembly system.
The requirement of integer strengths is simply one way of ``quantizing'' the minimum distinction we are willing to make between energies and then re-scaling so that this quantity is normalized to 1.\footnote{Indeed, our proof does not require that strengths be integer, merely that the distance between the smallest energy strong enough to bind and the largest energy too weak to bind be at least 1.}
Our 3-cooperative lower bound therefore shows that in general, certain self-assembling systems that have very large gaps between some of their binding energies nonetheless require other binding energies to be extremely close (exponentially small in terms of the larger gaps) and yet still unequal.
This can be interpreted as an infeasibility result if one defines ``exponentially fine control'' of binding energies as ``infeasible'' to execute in any real laboratory, since no implementation of the specified tile behavior can use courser energies.

As a converse to the exponential temperature lower bound stated above, we show that there is a polynomial-time algorithm that, given any tile system $\calT$ with $K$ tile types specified by its desired binding behavior, finds a temperature and glue strengths at most $2^{O(K)}$ that implement this behavior or reports that no such strengths exist.
This algorithm is used to show our main result, that there is a polynomial-time algorithm that, given an $n \times n$ square, determines the smallest tile assembly system (at any temperature) that uniquely self-assembles the square, answering the open question of \cite{ACGHKMR02}.

\section{Abstract Tile Assembly Model}
\label{sec-tam-informal}

This section gives a brief informal sketch of the abstract Tile Assembly Model (aTAM).
See Section \ref{sec-tam-formal} for a formal definition of the aTAM.

A \emph{tile type} is a unit square with four sides, each having a \emph{glue label} (often represented as a finite string).
We assume a finite set $T$ of tile types, but an infinite number of copies of each tile type, each copy referred to as a \emph{tile}.
An \emph{assembly}
(a.k.a., \emph{supertile})
is a positioning of tiles on (part of) the integer lattice $\Z^2$; i.e., a partial function $\Z^2 \dashrightarrow T$. 
For a set of tile types $T$, let $\Lambda(T)$ denote the set of all glue labels of tile types in $T$.
We may think of a tile type as a function $t:\dall\to\Lambda(T)$ indicating, for each direction $d\in\dall$ (``north, south, east, west''), the glue label $t(d)$ appearing on side $d$.
A \emph{strength function} is a function $g:\Lambda(T)\to\N$ indicating, for each glue label $\ell$, the strength $g(\ell)$ with which it binds.
Two adjacent tiles in an assembly \emph{interact} if the glue labels on their abutting sides are equal and have positive strength according to $g$.
Each assembly induces a \emph{binding graph}, a grid graph whose vertices are tiles, with an edge between two tiles if they interact.
The assembly is \emph{$\tau$-stable} if every cut of its binding graph has strength at least $\tau$, where the weight of an edge is the strength of the glue it represents.
That is, the assembly is stable if at least energy $\tau$ is required to separate the assembly into two parts.

A \emph{tile assembly system} (TAS) is a quadruple $\calT = (T,\sigma,g,\tau)$, where $T$ is a finite set of tile types, $\sigma:\Z^2 \dashrightarrow T$ is a finite, $\tau$-stable \emph{seed assembly}, $g:\Lambda(T) \to \N$ is a \emph{strength function}, and $\tau$ is a \emph{temperature}.
In this paper, we assume that all seed assemblies $\sigma$ consist of a single tile type (i.e., $|\dom\sigma|=1$).
Given a TAS $\calT = (T,\sigma,g,\tau)$, an assembly $\alpha$ is \emph{producible} if either $\alpha = \sigma$ or if $\beta$ is a producible assembly and $\alpha$ can be obtained from $\beta$ by placing a single tile type $t$ on empty space (a position $p\in\Z^2$ such that $\beta(p)$ is undefined), such that the resulting assembly $\alpha$ is $\tau$-stable.
In this case write $\beta\to_1 \alpha$ ($\alpha$ is producible from $\beta$ by the stable attachment of one tile), and write $\beta\to \alpha$ if $\beta \to_1^* \alpha$ ($\alpha$ is producible from $\beta$ by the stable attachment of zero or more tiles).
An assembly is \emph{terminal} if no tile can be stably attached to it.
Let $\prodasm{\calT}$ be the set of producible assemblies of $\calT$, and let $\termasm{\calT} \subseteq \prodasm{\calT}$ be the set of producible, terminal assemblies of $\calT$.
Given a connected shape $S \subseteq \Z^2$,  a TAS $\calT$ \emph{uniquely self-assembles $S$}  if $\termasm{\calT} = \{\ha\}$ and $\dom\ha = S$.

\section{Finding Strengths to Implement a Tile System}
\label{sec-find-strengths}

In this section we show that there is a polynomial-time algorithm that, given a desired behavior of a tile system, can find strengths to implement that behavior that are at most exponential in the number of tile types, or report that no such strengths exist.
This algorithm is the main technical tool used in the proof of our main result, Theorem~\ref{thm-min-tile-set-squares}.


First, we formalize what we mean by the ``behavior'' of a tile system.
Let $T$ be a set of tile types, and let $t \in T$.
Given a strength function $g:\Lambda(T) \to \N$ and a temperature $\tau \in \Z^+$, define the \emph{cooperation set of $t$ with respect to $g$ and $\tau$} to be the collection
$\mathcal{D}_{g,\tau}(t) = \setr{D \subseteq \dall}{\sum_{d \in D} g(t(d)) \geq \tau}$,
i.e., the collection of subsets of sides of $t$ whose glues have sufficient strength to bind cooperatively.
Let $\sigma:\Z^2 \dashrightarrow T$ be a seed assembly, let $\tau_1,\tau_2 \in \Z^+$ be temperatures, and let $g_1,g_2:\Lambda(T)\to\N$ be strength functions.
We say that the TAS's $\calT_1 = (T,\sigma,g_1,\tau_1)$ and $\calT_2 = (T,\sigma,g_2,\tau_2)$, differing only on their strength function and temperature,
are \emph{locally equivalent} if, for each tile type $t \in T$, $\mathcal{D}_{g_1,\tau_1}(t) = \mathcal{D}_{g_2,\tau_2}(t).$\footnote{Note that the definition of equivalence is independent of the seed assembly; we include it only to be able to talk about the equivalence of TAS's rather than the more cumbersome ``equivalence of triples of the form $(T,g,\tau)$.''}
The behavior of an individual tile type during assembly is completely determined by its cooperation set, in the sense that if $\calT_1$ and $\calT_2$ are locally equivalent, then $\prodasm{\calT_1} = \prodasm{\calT_2}$ and $\termasm{\calT_1} = \termasm{\calT_2}$.\footnote{The converse does not hold, however.  For instance, some tile types may have a subset of sides whose glues never appear together at a binding site during assembly, so it would be irrelevant to the definition of $\prodasm{\calT}$ and $\termasm{\calT}$ whether or not that combination of glues have enough strength to bind.}

Even without any strength function or temperature, by specifying a cooperation set for each tile type, one can describe a ``behavior'' of a TAS in the sense that its dynamic evolution can be simulated knowing only the cooperation set of each tile type.
We call a system thus specified a strength-free TAS.
More formally, a \emph{strength-free TAS} is a triple $(T, \sigma, \mathcal{D})$, where $T$ is a finite set of tile types, $\sigma: \Z^2 \dashrightarrow T$ is the size-1 seed assembly, and $\mathcal{D}:T \to \mathcal{P}(\mathcal{P}(\dall))$ is a function from a tile type $t \in T$ to a cooperation set $\mathcal{D}(t)$.
For a standard TAS $\calT = (T, \sigma, g, \tau)$ and a strength-free TAS $\calT_{\rm sf} = (T, \sigma, \mathcal{D})$ sharing the same tile set and seed assembly, we say that $\calT$ and $\calT_{\rm sf}$ are \emph{locally equivalent} if $\mathcal{D}_{g, \tau}(t) = \mathcal{D}(t)$ for each tile type $t \in T$.

Note that every TAS has a unique cooperation set for each tile type, and hence, has a locally equivalent strength-free TAS.
Say that a strength-free TAS is \emph{implementable} if there exists a TAS locally equivalent to it.
Not every strength-free TAS is implementable.
This is because cooperation sets could be contradictory; for instance, two tile types $t_1$ and $t_2$ could satisfy $t_1(\dN)=t_2(\dN)$ and $\{\dN\} \in \mathcal{D}(t_1)$ but $\{\dN\} \not\in \mathcal{D}(t_2)$.




\begin{theorem}\label{thm-sf-to-standard}
There is a polynomial time algorithm that, given a strength-free TAS $\calT_{\rm sf}=(T,\sigma,\mathcal{D})$, determines whether there exists a locally equivalent TAS $\calT=(T,\sigma,g,\tau)$ and outputs such a $\calT$ if it exists.
Furthermore, it is guaranteed that $\tau \leq 2^{O(|T|)}$.
\end{theorem}


\newcommand{\ProofFindLocallyEquiv}{
    Intuitively, we count the number of different total behaviors a tile can have, based on equating its behavior with its cooperation set.
    We then cast the problem of finding strengths to implement a given strength-free tile system as finding solutions to a certain system of linear inequalities, which is solved by Gaussian elimination.
    We then argue that the vertices of the polytope defined by this system have rational numbers with numerator and denominator that are at most exponential in the number of inequalities, which is itself linear in the number of tile types.
    This implies that by multiplying the rational numbers at one of these vertices by their least common multiple, which preserves the inequalities, we obtain an integer solution with values at most exponential in the number of tile types.

    Formally, let $\calT_{\rm sf}$ be a strength-free TAS with a tile set $T$ of $k$ tile types with $u \leq k+1$ different glues.\footnote{Each tile type has 4 sides so it might seem that there could be $4k$ total glues if there are $k$ tile types.  However, in a nontrivial system (one that has no ``effectively null'' glues that appear on only one side of any tile type), for each side of a tile type, the choice of glue for that side is limited to those glues on the opposite side of the $k$ tile types, or alternately we could choose the null strength-0 glue.}
    We would like to decide whether $\calT_{\rm sf}$ is in fact realizable by a TAS.
    To have the tightest upper bound on the temperature, ideally we could to solve the problem of finding the minimum temperature TAS that is locally equivalent to $\calT_{\rm sf}$.
    This optimization problem can be cast as an integer linear program on a temperature variable $\tau$ and a set of glue-strength variables $s_1, s_2, \ldots, s_u$ as in the following example:

    \[
    \begin{array}{llll}
    \textrm{Minimize}	& \tau  &   &    \\
    \textrm{subject to}	& \tau, s_1, s_2, \ldots, s_u &\in& \N \\
                     	& s_1 + s_3 - \tau   & \geq & 0  \\
                     	& s_1 + s_4 + s_6 - \tau    & \geq & 0 \\
                     	& &\ldots& \\
                     	& s_1 + s_2 + s_4 - \tau    & \leq &  -1 \\
                     	& s_2 + 2 s_3 - \tau    & \leq & -1 \\
                     	& &\ldots&
    \end{array}
    \]

    The ``$\geq 0$'' inequalities correspond to the union of all cooperation sets of all tile types, and the ``$\leq -1$'' inequalities correspond to the union of all complements of the cooperation sets; i.e., each set $\mathcal{P}(\dall) \setminus \mathcal{D}(t)$, where $\mathcal{D}(t)$ is the cooperation set of $t$.
    Since we require each strength to be an integer, ``$\leq -1$'' is equivalent to ``$< 0$''.
    Since each tile type $t$ has $|\mathcal{D}(t)|$ ``$\geq 0$'' inequalities (one for each subset of sides in its cooperation set) and $16 - |\mathcal{D}(t)|$ ``$\leq -1$'' inequalities, there are $16k$ inequalities in the integer linear program.

    However, solving this optimization problem is not necessary to prove the theorem.
    Our goal will not be to find the smallest temperature TAS that satisfies the constraints above (which remains an open problem to do in polynomial time), but simply to find any feasible integer solution with temperature and strengths at most $2^{O(k)}$.
    Call the above system of constraints (including the integer constraint) $S_1$.
    Consider the real-valued system of linear inequalities $S_2$ defined as the above inequalities with the integer constraint $\tau, s_1, s_2, \ldots, s_u \in \N$ relaxed to simply $\tau, s_1, s_2, \ldots, s_u \geq 0$.
    Then we have the implication ``$S_1$ has a solution'' $\implies$ ``$S_2$ has a solution''.
    Conversely, any rational-valued solution to $S_2$ can be converted to an integer-valued solution to $S_1$ by multiplying each value by the least common multiple $L$ of the denominators of the rational numbers.\footnote{This actually enforces the stronger condition that each ``$\leq -1$'' inequality is actually ``$\leq -L$''.  This is possible because we have no \emph{upper} bound on the variables, which would prevent multiplication from necessarily preserving the inequalities.
    }
    Furthermore, since the input coefficients are integers, if the feasible polytope of $S_2$ is non-empty, then all of its vertices are rational.
    Therefore $S_2$ has a solution if and only if it has an integer solution.
    Since any integer solution to $S_2$ is a solution to $S_1$, we have the full bidirectional implication ``$S_1$ has a solution'' $\iff$ ``$S_2$ has an integer solution''.
    We can pick any $n$ linearly independent inequalities of $S_2$, interpret them as equalities, and use Gaussian elimination (with exact rational arithmetic) to obtain some vertex of the feasible polytope described by the inequalities, and convert these to integer solutions to $S_1$ through multiplication as described above.
    If we cannot find $n$ linearly independent inequalities (testable by computing the rank of the matrix defining the inequalities) then there is no TAS implementing the behavior of $\calT_\mathrm{sf}$.
    It remains to show that in case there is a solution, the integers we obtain by this method obey the stated upper bound $2^{O(k)}$.


    Each coefficient has absolute value at most 2 (since we may assume N/S glues are disjoint from E/W glues), and each equation has at most 5 nonzero left side terms since each tile type has only 4 sides (together with the $-1$ coefficient for $\tau$).
    Applying Lemma~\ref{lem-rational-small-num-den} (stated and proven after the current proof) with $n=u+1$, $c_1 = 2$, and $c_2 = 5$, we obtain that each vertex is a rational vector $\vec{x}=(\frac{p_1}{q_1},\ldots,\frac{p_{u+1}}{q_{u+1}})$ such that, for each $1 \leq i \leq u+1$,  $|p_i|,|q_i| \leq 2^{u+1} 6^{(u+1) / 2} = 2^{u + 1 + ((u + 1) / 2) \cdot \log 6} < 2^{3u+3}$.
    Since we enforce nonnegativity, $p_i=|p_i|$ and $q_i=|q_i|$.
    We multiply the rational vector $\vec{x}$ by $L = \mathrm{LCM}(q_1,\ldots,q_{u+1})$ to obtain the integer vector $\vec{x}'$.
    By Lemma~\ref{lem-rational-small-num-den}, $L \leq 2^{3u+3}$.
    Then each integer solution value $x_i'$ obeys $x_i' = L \cdot x_i \leq (2^{3u+3})^2 \leq 2^{2(3(k+1)+3)} = 2^{O(k)}$.
}

\begin{proof}
\ProofFindLocallyEquiv
\qedl\end{proof}

\newcommand{\LemmaRationalSmallNumDen}{

    \renewcommand{\vb}{\vec{b}}

    \begin{lem} \label{lem-rational-small-num-den}
      Let $c_1,c_2 \in \Z^+$ be constants, and let $\vb$ be an $n \times 1$ integer column vector and $A=(a_{ij})$ be a nonsingular $n \times n$ integer matrix such that for each $i,j$, $|a_{ij}| \leq c_1$ and $|b_j| \leq c_1$, and each row of $A$ contains at most $c_2$ nonzero entries.
      Then the solution to the linear system $A\vec{x} = \vb$ is a rational vector $\vec{x}\in\Q^n$ such that, if each component $x_i=\frac{p_i}{q_i}$ is written in lowest terms with $p_i,q_i\in\Z$, then $|p_i| \leq c_1^n (c_2+1)^{n/2}$ and $|q_i| \leq c_1^n c_2^{n/2}$.
      Furthermore, the least common multiple of all the $q_i$'s is at most $c_1^n c_2^{n/2}$.
    \end{lem}

    \begin{proof}
      Recall Hadamard's inequality $|\det A| \leq \prod_{i=1}^n \| v_i \|_2$, where $v_i$ is the $i^\text{th}$ row of $A$.\footnote{Hadamard's inequality is typically stated for $v_i$ a column of $A$, but the determinant of a matrix and its transpose are equal so the bound holds when taking the product over rows as well.}
      Since $v_i$ has at most $c_2$ nonzero entries that are each at most absolute value $c_1$, Hadamard's inequality tells us that $|\det A| \leq \prod_{i=1}^n \sqrt{c_2 \cdot c_1^2} = c_1^n c_2^{n/2}$.
      Similarly, letting $A_i$ be $A$ with column vector $\vec{b}$ replacing $A$'s $i^\text{th}$ column, $A_i$ has at most $c_2+1$ nonzero entries per row, so a similar argument gives $|\det A_i| \leq c_1^n (c_2+1)^{n/2}$.
      The $i^\text{th}$ solution is $x_i = \frac{\det A_i}{\det A}$ by Cramer's rule.
      Since $A$ and $A_i$ are integer-valued, so are $\det A$ and $\det A_i$, whence the upper bounds on $|\det A|$ and $|\det A_i|$ also apply to $|q_i|$ and $|p_i|$, respectively, since they are $x_i$'s lowest terms representation.
      Since the lowest terms representation of each $q_i$ necessarily divides $\det A$, this implies that their least common multiple is also at most $c_1^n c_2^{n/2}$.
    \qedl\end{proof}
}

\LemmaRationalSmallNumDen

\section{Finding the Minimum Tile Assembly System Assembling a Square at any Temperature}
\label{sec-min-tiles-square}

\newcommand{\DMTSS}{\textsc{MinDirectedTileSetSquare}}
\newcommand{\Kdtc}{\mathrm{C}^\mathrm{dtc}}
\newcommand{\US}{\textsc{Unique-Shape}}

This section proves the main result of this paper.
We show that there is a polynomial-time algorithm that, given an $n \times n$ square $S_n$, computes the smallest TAS that uniquely self-assembles $S_n$.
Adleman, Cheng, Goel, Huang, Kempe, Moisset de Espan\'{e}s, and Rothemund~\cite{ACGHKMR02} showed that this problem is polynomial-time solvable when the temperature is restricted to be 2, and asked whether there is an algorithm that works when the temperature is unrestricted, which we answer affirmatively.


The next proposition is useful in the proof of our main theorem.

\begin{prop}\label{prop-num-TAS}
  For each $k\in\Z^+$, there are at most $168^k k^{4k+2}$ implementable strength-free TAS's with at most $k$ tile types.
\end{prop}

\begin{proof}
    If $\calT=(T,\sigma,\mathcal{D})$ is an implementable strength-free TAS, then for all $t \in T$, the cooperation set $\mathcal{D}(t)$ of $t$ is a collection of subsets of $\mathcal{\dall}$ that is closed under the superset operation, i.e., if $D \subseteq D' \subseteq \dall$ and $D \in \mathcal{D}(t)$, then $D' \in \mathcal{D}(t)$.
    This closure property is due to the fact that having strictly more sides available to bind cannot inhibit binding that would otherwise occur, as long as strengths are assumed to be nonnegative.
    Each cooperation set $\mathcal{D}$ is defined by a unique \emph{antichain}, which is a subcollection $\mathcal{D}' = \{D'_1, \ldots, D'_m\} \subseteq \mathcal{D}$ such that, for all $1 \leq i,j \leq m$, $D'_i \not\subseteq D'_j$, whose closure under the superset operation is equal to $\mathcal{D}$.
    The antichain consists of the minimal elements of $\mathcal{D}$ under the partial order $\subseteq$.
    The number of antichains of subsets of $\dall$ is given by the fourth Dedekind number $M(4) = 168$~\cite{oeisA000372}.
    Thus, a tile type has at most 168 different possible cooperation sets.

    For each side of a tile type, there are at most $k$ glue labels to choose, so there are at most $k^4$ ways to assign these labels to each side.
    Therefore, encoding each tile type as a list of 4 glue labels and cooperation set, and encoding a TAS as a list of tile types, there are at most $(168 k^4)^k = 168^k k^{4k}$ different strength-free tile sets with $k$ tile types.
    Since there are $k$ choices for the seed tile, there are at most $168^k k^{4k+1}$ different strength-free TAS's with $k$ tile types.
    Thus there are at most $\sum_{i=1}^k 168^i i^{4i+1} \leq 168^k k^{4k+2}$ such TAS's with \emph{at most} $k$ tile types.
\qedl\end{proof}

The following theorem is the main result of this paper.

\begin{theorem} \label{thm-min-tile-set-squares}
  There is a polynomial-time algorithm that, given an $n \times n$ square $S_n$, outputs a minimal TAS $\calT$ that uniquely self-assembles $S_n$.
\end{theorem}

\newcommand{\ProofMinTileSetSquares}{
%
    In~\cite{ACGHKMR02}, the authors study the variant of the minimum tile set problem restricted to squares where the temperature is fixed at $\tau=2$.
    They use the following argument to show the problem is solvable in polynomial time.
    For all $n \in \N$, let $S_n = \{0,1,\ldots,n-1\}^2$ denote the $n \times n$ square.
    Adleman, Cheng, Goel, and Huang~\cite{AdChGoHu01} showed that for all $n \in \N$, there is a TAS with at most $O(\frac{\log n}{\log \log n})$ tile types that uniquely self-assembles $S_n$.
    The proof of~\cite{ACGHKMR02} first shows by a simple counting argument that there are at most a polynomial in $n$ number of temperature-2 TAS's with at most $O(\frac{\log n}{\log \log n})$ tile types, using the fact that all strengths may without loss of generality be assumed to be 0, 1, or 2.
    They then make use of a polynomial-time algorithm $\US$ devised in the same paper~\cite{ACGHKMR02} that, given any shape $S$ and any TAS $\calT$, determines whether $\calT$ uniquely self-assembles $S$. 
    Finding the minimum TAS for $S_n$ then amounts to iterating over every ``small'' ($O(\frac{\log n}{\log \log n})$ tile types) TAS $\calT$ and using the algorithm $\US$ to check which of these systems assemble $S_n$.
    The upper bound of~\cite{AdChGoHu01} guarantees that $\US$ will report a positive answer for at least one of these systems.

    Let $c \in \Z^+$ be the constant, shown to exist in~\cite{AdChGoHu01}, such that, for all $n\in\N$, some tile system with at most $c \log n / \log \log n$ tile types uniquely self-assembles $S_n$.
    Let $k = c \log n / \log \log n$.
    Enumerate each strength-free TAS with at most $k$ tile types and test whether the strength-free system uniquely self-assembles $S_n$.
    The algorithm $\US$ of~\cite{ACGHKMR02} can be executed just as easily with a strength-free TAS as with a standard TAS, so $\US$ may be used for this test.

    Proposition \ref{prop-num-TAS} tells us that the number of strength-free TAS's with at most $k$ tile types is at most $168^k k^{4k+2}$.
    We have that
    $
        168^k
        \leq
        168^{c \log n / \log \log n}
        \leq
        168^{c \log n}
        \leq
        (2^{\log n})^{c \log 168}
        \leq
        n^{8 c},
    $
    and
    \begin{eqnarray*}
      k^{4k+2}
      &\leq&
      k^2 (c \log n / \log \log n)^{4 c \log n / \log \log n}
      \\&=&
      k^2 (2^{\log (c \log n / \log \log n)})^{4 c \log n / \log \log n}
      \\&=&
      k^2 (2^{\log c + \log \log n - \log \log \log n})^{4 c \log n / \log \log n}
      \\&=&
      k^2 2^{4 c \log n (\log c + \log \log n - \log \log \log n) / \log \log n}
      \\&=&
      k^2 (2^{\log n})^{4 c (\log c + \log \log n - \log \log \log n) / \log \log n}
      \\&=&
      k^2 n^{4 c \log c / \log \log n + 4 c - 4 c \log \log \log n / \log \log n}
      \\&\leq&
      n^{4 c \log c + 4 c + 2},
    \end{eqnarray*}
    whence the number of strength-free TAS's to search is at most $n^{4 c \log c + 12 c + 2} = \mathrm{poly}(n)$.
    Once we obtain a smallest such strength-free TAS $\calT_\mathrm{sf}$, we run the algorithm described in Theorem~\ref{thm-sf-to-standard} to obtain a minimal standard TAS $\calT$ that uniquely self-assembles $S_n$.
}

    \begin{proof}
    \ProofMinTileSetSquares
    \qedl\end{proof}


    We note that while we have stated the theorem for the family of square shapes, our method, as well as that of~\cite{ACGHKMR02}, works for any family of shapes $S_1, S_2, \ldots$ where $|S_n| = \mathrm{poly}(n)$ and the tile complexity of $S_n$ is at most $O(\frac{\log n}{\log \log n})$.
    This includes, for instance, the family $\{ T_1, T_2, \ldots \}$, where $T_n$ is a width-$n$ right triangle, and for each $q \in \Q^+$ the family $\{ R_{q,1}, R_{q,2},\ldots \}$, where $R_{q,n}$ is the $n \times \floor{qn}$ rectangle.




\section{Bounds on Temperature Relative to Number of Tile Types}

This section shows two bounds relating  the number of tile types in a TAS to its temperature.
The first bound, Theorem~\ref{thm-exponential-tiles}, shows that there are TAS's that require temperature \emph{exponential} in the number of tile types (in the sense of local equivalence as defined in Section~\ref{sec-find-strengths}), if any combination of sides may be used for binding.
This result can be interpreted to mean that the algorithm of~\cite{AdChGoHu01} to find the minimum temperature-2 TAS for assembling an $n \times n$ square, which searches over all possible assignments of strengths to the glues, cannot be extended in a straightforward manner to handle larger temperatures, which is why it is necessary for the algorithm of Theorem~\ref{thm-min-tile-set-squares} to ``shortcut'' through the behaviors of tile types rather than enumerating strengths.
The second bound, Theorem~\ref{thm-linear-temperature}, on the other hand, shows that if we restrict attention to those (quite prevalent) classes of tile systems that use only one or two sides of tiles to bind, then \emph{linear} temperature always suffices.

\subsection{Tile Assembly Systems Requiring Temperature Exponential in Number of Tile Types}
\label{sec-exponential-tiles}

In this section, we prove that a temperature that is exponential in the number of tile types given by Theorem~\ref{thm-sf-to-standard} is optimal, although there is a gap between the exponents ($2^{|T|/4}$ for Theorem~\ref{thm-exponential-tiles} below versus $O(2^{6 |T|})$ for Theorem~\ref{thm-sf-to-standard}).

\newcommand{\FigExponentialTiles}{
\begin{figure}[htb]
\begin{center}
  \includegraphics[width=\opt{normal}{4}\opt{submission}{3.0}in]{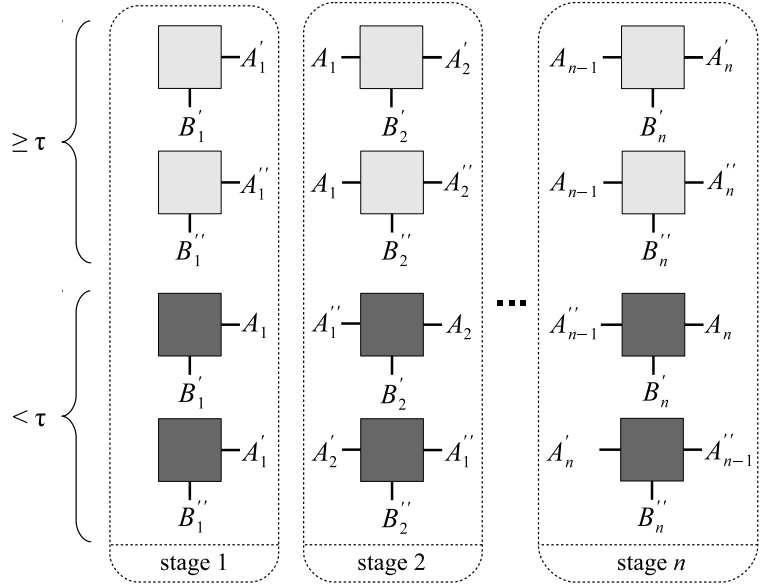}
  \caption{\label{fig:exponential-tiles-example} \figuresize
  A set of tile types requiring temperature that is exponentially larger than the number of tile types.
  There are stages $1, 2, \ldots , n$, with stage $i$ containing 4 tiles, and stage $i$ ensuring that the gap between the largest and smallest strength in the stage is at least $2^i$.
  In each stage, each of the top two light tiles represents a triple (a pair in stage 1) of glues whose sum is at least $\tau$.
  Each of the bottom two dark tiles represents a triple of glues whose sum is less than $\tau$.
  The inequalities are satisfiable, for instance, by setting $A_n = 3^{n-1}, A'_n = 2 A_n, A''_n = 3 A_n, B_n = \tau - A_n, B'_n = \tau - A'_n, B''_n = \tau - A''_n$.
  }
\end{center}
\end{figure}
}

\FigExponentialTiles

\begin{theorem} \label{thm-exponential-tiles}
  For every $n\in\Z^+$, there is a TAS $\calT = (T,\sigma,g,\tau)$ such that $|T| = 4n$ and for every TAS $\calT' = (T,\sigma,g',\tau')$ that is locally equivalent to $\calT$, $\tau' \geq 2^{n}.$
\end{theorem}

\newcommand{\ProofExponentialTiles}{
  The tile set $T$ is shown in Figure~\ref{fig:exponential-tiles-example}.\footnote{We do not specify the seed assembly since we are concerned only with the local behavior of the tiles.  To make the local equivalence nontrivial, we would need to add a small number of tile types to the TAS to ensure that each tile shown is actually attachable at some point during assembly.  However, but this would not affect the asymptotic size of the tile set as $n\to\infty$, so the exponential lower bound on the temperature would still hold.}
  In each stage, each of the top two light tile types represents a triple (a pair in stage 1) of glues whose sum is at least $\tau$.
  Each of the bottom two dark tiles represents a triple of glues whose sum is less than $\tau$.
  For the dark tile types to be nontrivial, we could imagine that the (unlabeled) north glue is strong enough to cooperate with some of the other glues.
  The actual strengths of the glues are left as variables, but the caption of Figure~\ref{fig:exponential-tiles-example} gives one example of strengths that would satisfy the inequalities that the tile types represent.

  We prove by induction on $n$ that $A''_n \geq A_n + 2^n$.
  For the base case, in the first stage, the top light tile type and top dark tile type enforce that
  $
    A'_1 + B'_1 \geq \tau > A_1 + B'_1,
  $
  so $A_1 < A'_1$.
  Similarly, the bottom light tile type and the bottom dark tile type enforce that $A'_1 < A''_1$.
  Therefore $A''_1 \geq A_1 + 2$.

  For the inductive case, assume that $A''_{n-1} \geq A_{n-1} + 2^{n-1}$.
  The top dark tile type enforces that
  $
    \tau > A''_{n-1} + A_n + B'_n,
  $
  and by the induction hypothesis,
  $
    A''_{n-1} + A_n + B'_n
    \geq
    A_{n-1} + 2^{n-1} + A_n + B'_n.
  $
  The top light tile type enforces that
  $
    A'_n + A_{n-1} + B'_n \geq \tau,
  $
  which combined with the previous two inequalities shows that
  $
    A'_n \geq 2^{n-1} + A_n.
  $
  A similar analysis with the bottom light tile type and bottom dark tile type shows that
  $
    A''_n \geq 2^{n-1} + A'_n,
  $
  whence
  $
    A''_n \geq  2^{n-1} + 2^{n-1} + A_n,
  $
  establishing the inductive case.

  Since it can be assumed without loss of generality that strengths are at most $\tau$, this shows that the tile set consisting of $n$ stages, having $|T|=4n$ tile types, requires $\tau \geq 2^n$ to be realized.
}

    \begin{proof}
    \ProofExponentialTiles
    \qedl\end{proof}


\subsection{Temperature Linear in the Number of Tile Types Suffices for 2-Cooperative Equivalence}

Theorem~\ref{thm-exponential-tiles} shows that temperature exponentially larger than the number of tile types is sometimes necessary for a TAS's behavior to be realized by integer strengths.
However, the definition of local equivalence assumes that all possible combinations of sides of a tile type may be present in an assembly.
Many TAS's are more constrained than this.
There is a wide class of TAS's that we term \emph{2-cooperative}, meaning that all binding events during all assembly sequences use only 1 or 2 sides that bind with positive strength.
Nearly all theoretical TAS's found in the literature are 2-cooperative (indeed, temperature 2 systems by definition cannot \emph{require} three sides to be present, although the model allows tile attachments with excess strength).
In this section we show that the 3-cooperativity of Figure~\ref{fig:exponential-tiles-example} is necessary, by showing that 2-cooperative systems can always be realized by strengths linear in the number of tile types.

\begin{theorem} \label{thm-linear-temperature}
  Let $\calT = (T,\sigma,g,\tau)$ be a TAS, and let $\mathcal{D}_{g,\tau}^{(2)}(t) \subseteq \mathcal{D}_{g,t}(t)$ be the cooperation set of $t$ with respect to $g$ and $\tau$ restricted to containing only subsets of $\dall$ of cardinality 1 or 2.  Then there is a TAS $\calT' = (T,\sigma,g',\tau')$ with $\tau' \leq 2 |T| + 2$ such that, for each $t \in T$, $\mathcal{D}^{(2)}_{g,\tau}(t) = \mathcal{D}^{(2)}_{g',\tau'}(t)$.
\end{theorem}

That is, $\calT'$ is equivalent in behavior to $\calT$, so long as all attachments involve only 1 or 2 sides.

\newcommand{\ProofTwoCooperative}{
  Let $K = |T|$.
  Let $\Lambda(T)$ denote the set of all glue labels on tile types in $T$.
  Let $G = \setr{g(\sigma)}{\sigma \in \Lambda(T)} \setminus \{0,\tau\}$.
  That is, $G$ is the set of all positive but insufficient glue strengths used in this system.
  $|G| \leq 2 K$ since there are at most $|T|$ north-south glues and at most $|T|$ east-west glues.

  We split $G$ into two subsets $L = \setr{g \in G}{0 < g < \tau/2}$ and 
  $H = \{\ g \in G\ |\ \tau/2 \leq g < \tau\ \}$.
  Let $L = \{\ell_1, \ell_2, \ldots, \ell_n\}$ such that $\ell_1 < \ell_2 < \ldots < \ell_n < \tau/2$, and let $H = \{h_1, h_2, \ldots, h_m \}$ with $\tau/2 \leq h_1 < h_2 \ldots h_m < \tau.$
  For descriptive purposes, define $\ell_0 = 0$ and $\ell_{n+1} = \tau/2$ (although these numbers may not be glue strengths).

  We aim at designing an algorithm to find a glue function $g'$ satisfying:
  \begin{enumerate}
    \item \label{cond-1-new-strength} for any label $\sigma \in \{ \sigma_1, \ldots, \sigma_{2K} \}, g(\sigma) \geq \tau \iff g'(\sigma) \geq 2n + 2;$

    \item \label{cond-2-new-strength} for any pair of labels $\{\sigma, \sigma'\} \subset \{ \sigma_1, \ldots, \sigma_{2K} \}, g(\sigma) + g(\sigma') \geq \tau \iff g'(\sigma) + g'(\sigma') \geq 2n + 2.$
  \end{enumerate}
  Then for $\tau' = 2n + 2$, the TAS $\calT' = (T,\sigma,g',\tau')$ satisfies the 2-cooperative equivalence with $\calT$ that we seek.

  First, we define an equivalence relation $\equiv$ on $H$ defined as: for $h,h'\in H$, $h \equiv h'$ if $(\forall\ 1 \leq i \leq n) (h + \ell_i \geq \tau \iff h' + \ell_i \geq \tau).$
  This partitions $H$ into subsets $H_{n+1}, \ldots, H_2, H_1$ such that $h \in H_j$ if and only if $\tau - \ell_j \leq h < \tau = \ell_{j-1}.$
  In other words, $h \in H_j$ if $h + \ell_j \geq \tau$ implies $i \geq j$.

  A glue function $g':\Lambda(T) \to \N$ is defined as follows: for a label $\sigma \in \Lambda(T),$
  $$
    g'(\sigma) = \left\{
      \begin{array}{ll}
        0, & \hbox{if $g(\sigma)=0$;} \\
        i, & \hbox{if $g(\sigma)=\ell_i$ ($1 \leq i \leq n$);} \\
        2n+2-j, & \hbox{if $g(\sigma)\in H_j$ ($1 \leq j \leq n+1$);} \\
        2n+2, & \hbox{if $\tau \leq g(\sigma)$.}
      \end{array}
    \right.
  $$

  It is trivial that this satisfies condition \eqref{cond-1-new-strength} above.
  Let $\sigma,\sigma' \in \Lambda(T)$ with $0 < g(\sigma) \leq g(\sigma') < \tau.$
  There are two cases.
\begin{enumerate}
  \item Suppose that $\tau \leq g(\sigma) + g(\sigma')$.
  Note that $\tau/2 \leq g(\sigma')$.
If $\tau/2 \leq g(\sigma)$, then by definition $n+1 \leq g'(\sigma)$ and $n+1 \leq g'(\sigma')$, whence their sum it at least $2n+2$.
If $g(\sigma) < \tau/2$, then let $g(\sigma) = \ell_i$ for some $1 \leq i \leq n$.
Since $\tau \leq g(\sigma) + g(\sigma')$, this means that $\tau = \ell_i \leq g(\sigma')$, which in turn implies $g(\sigma') \in H_i \cup H_{i-1} \cup \ldots \cup H_1$.
By definition, $g'(\sigma) = i$ and $g'(\sigma') \geq 2n + 2 - i$.
Consequently, $2n+2 \leq g'(\sigma) + g'(\sigma')$.

  \item Suppose that $g(\sigma) + g(\sigma') < \tau$.
In this case, $g(\sigma) < \tau/2$.
If $g(\sigma') < \tau/2$, then by the definition of $g'$, both $g(\sigma)$ and $g(\sigma')$ are at most $n$ so that their sum cannot reach $2n+2$.
Otherwise, let $g(\sigma) = \ell_i$.
An argument similar to the one above gives $g(\sigma') \in H_{n+1} \cup \ldots \cup H_{i+1}$, whence $g'(\sigma') < 2n + 2 - i$.
Thus, $g'(\sigma) + g'(\sigma') < 2n + 2$.
\end{enumerate}
This verifies that $g'$ satisfies condition \eqref{cond-2-new-strength}.
}

\begin{proof}
\ProofTwoCooperative
\qedl\end{proof}

\section{Open Questions}
\label{sec-conclusion}

%
%
%
%
Our polynomial-time algorithm for finding the minimal tile system to self-assemble a square made crucial use of our polynomial-time algorithm that, given a strength-free tile system $\calT_\mathsf{sf}$, finds strengths and a temperature to implement a locally equivalent TAS $\calT$, or reports that none exists.
An open question is whether there is a polynomial-time algorithm that, given a strength-free TAS $\calT_\mathsf{sf}$, outputs a TAS \emph{of minimal temperature} that is locally equivalent to $\calT_\mathsf{sf}$, or reports that none exists.
In the proof of Theorem~\ref{thm-sf-to-standard}, we expressed the problem as an integer linear program, but needed only to find an integer feasible solution that is not necessarily optimal.
If we instead wanted to find an \emph{optimal} solution to the program, it is not clear whether the problem restricted to such instances is $\NP$-hard.

The next question is less formal.
Our results relating to 3-cooperative and 2-cooperative systems (Theorems~\ref{thm-exponential-tiles} and~\ref{thm-linear-temperature}, respectively), show that there is a difference in self-assembly ``power'' between these two classes of systems when we consider two tile systems to ``behave the same'' if and only if they are locally equivalent, which is a quite strict notion of behavioral equivalence.
For example, two tile systems could uniquely self-assemble the same shape even though they have different tile types (hence could not possibly be locally equivalent). 
It would be interesting to know to what extent 3-cooperative systems -- or high temperature tile systems in general -- are strictly more powerful than temperature 2 systems under more relaxed notions of equivalence.
For example, it is not difficult to design a shape that can be uniquely self-assembled with slightly fewer tile types at temperature 3 than at temperature 2.
How far apart can this gap be pushed?
Are there shapes with temperature-2 tile complexity that is ``significantly'' greater than their absolute (i.e., unrestricted temperature) tile complexity?



\paragraph{Acknowledgement.}
The authors thank Ehsan Chiniforooshan especially, for many fruitful and illuminating discussions that led to the results on temperature, and also Adam Marblestone and the members of Erik Winfree's group, particularly David Soloveichik, Joe Schaeffer, Damien Woods, and Erik Winfree, for insightful discussion and comments.
The second author is grateful to Aaron Meyerowitz (via the website \url{http://mathoverflow.net}) for pointing out the Dedekind numbers as a way to count the number of collections of subsets of a given set that are closed under the superset operation.

\opt{submission}{\newpage}
{ 
\bibliographystyle{plain}
\bibliography{tam}
}

    \appendix

\section{Appendix: Formal Definition of the Abstract Tile Assembly Model}
\label{sec-tam-formal}

\newcommand{\fullgridgraph}{G^\mathrm{f}}
\newcommand{\bindinggraph}{G^\mathrm{b}}

This section gives a terse definition of the abstract Tile Assembly Model (aTAM,~\cite{Winf98}). This is not a tutorial; for readers unfamiliar with the aTAM,  \cite{RotWin00} gives an excellent introduction to the model.

Fix an alphabet $\Sigma$. $\Sigma^*$ is the set of finite strings over $\Sigma$.
Given a discrete object $O$, $\langle O \rangle$ denotes a standard encoding of $O$ as an element of $\Sigma^*$.
$\Z$, $\Z^+$, $\N$, $\R^+$ denote the set of integers, positive integers, nonnegative integers, and nonnegative real numbers, respectively.
For a set $A$, $\calP(A)$ denotes the power set of $A$.
Given $A \subseteq \Z^2$, the \emph{full grid graph} of $A$ is the undirected graph $\fullgridgraph_A=(V,E)$, where $V=A$, and for all $u,v\in V$, $\{u,v\} \in E \iff \| u-v\|_2 = 1$; i.e., if and only if $u$ and $v$ are adjacent on the integer Cartesian plane.
A \emph{shape} is a set $S \subseteq \Z^2$ such that $\fullgridgraph_S$ is connected.

A \emph{tile type} is a tuple $t \in (\Sigma^*)^4$; i.e., a unit square with four sides listed in some standardized order, each side having a \emph{glue label} (a.k.a. \emph{glue}) $\ell \in \Sigma^*$.
For a set of tile types $T$, let $\Lambda(T) \subset \Sigma^*$ denote the set of all glue labels of tile types in $T$.
Let $\dall$ denote the \emph{directions} consisting of unit vectors $\{(0,1), (0,-1), (1,0), (-1,0)\}$.
Given a tile type $t$ and a direction $d \in \dall$, $t(d) \in \Lambda(T)$ denotes the glue label on $t$ in direction $d$.
We assume a finite set $T$ of tile types, but an infinite number of copies of each tile type, each copy referred to as a \emph{tile}. An \emph{assembly}
is a nonempty connected arrangement of tiles on the integer lattice $\Z^2$, i.e., a partial function $\alpha:\Z^2 \dashrightarrow T$ such that $\fullgridgraph_{\dom \alpha}$ is connected and $\dom \alpha \neq \emptyset$.
The \emph{shape of $\alpha$} is $\dom \alpha$.
Given two assemblies $\alpha,\beta:\Z^2 \dashrightarrow T$, we say $\alpha$ is a \emph{subassembly} of $\beta$, and we write $\alpha \sqsubseteq \beta$, if $\dom \alpha \subseteq \dom \beta$ and, for all points $p \in \dom \alpha$, $\alpha(p) = \beta(p)$.


A \emph{strength function} is a function $g:\Lambda(T)\to\N$ indicating, for each glue label $\ell$, the strength $g(\ell)$ with which it binds.
Let $\alpha$ be an assembly and let $p\in\dom\alpha$ and $d\in\dall$ such that $p + d \in \dom\alpha$.
Let $t=\alpha(p)$ and $t' = \alpha(p+d)$.
We say that the tiles $t$ and $t'$ at positions $p$ and $p+d$ \emph{interact} if $t(d) = t'(-d)$ and $g(t(d)) > 0$, i.e., if the glue labels on their abutting sides are equal and have positive strength.
Each assembly $\alpha$ induces a \emph{binding graph} $\bindinggraph_\alpha$, a grid graph $G=(V_\alpha,E_\alpha)$, where $V_\alpha=\dom\alpha$, and $\{p_1,p_2\} \in E_\alpha \iff \alpha(p_1) \text{ interacts with } \alpha(p_2)$.\footnote{For $\fullgridgraph_{\dom \alpha}=(V_{\dom \alpha},E_{\dom \alpha})$ and $\bindinggraph_\alpha=(V_\alpha,E_\alpha)$, $\bindinggraph_\alpha$ is a spanning subgraph of $\fullgridgraph_{\dom \alpha}$: $V_\alpha = V_{\dom \alpha}$ and $E_\alpha \subseteq E_{\dom \alpha}$.}
Given $\tau\in\Z^+$, $\alpha$ is \emph{$\tau$-stable} if every cut of $\bindinggraph_\alpha$ has weight at least $\tau$, where the weight of an edge is the strength of the glue it represents.
That is, $\alpha$ is $\tau$-stable if at least energy $\tau$ is required to separate $\alpha$ into two parts.
When $\tau$ is clear from context, we say $\alpha$ is \emph{stable}.

A \emph{tile assembly system} (TAS) is a triple $\calT = (T,\sigma,g,\tau)$, where $T$ is a finite set of tile types, $\sigma:\Z^2 \dashrightarrow T$ is the finite, $\tau$-stable \emph{seed assembly}, $g:\Lambda(T)\to\N$ is the \emph{strength function}, and $\tau\in\Z^+$ is the \emph{temperature}.
Given two $\tau$-stable assemblies $\alpha,\beta:\Z^2 \dashrightarrow T$, we write $\alpha \to_1^{\calT} \beta$ if $\alpha \sqsubseteq \beta$ and $|\dom \beta \setminus \dom \alpha| = 1$. In this case we say $\alpha$ \emph{$\calT$-produces $\beta$ in one step}.\footnote{Intuitively $\alpha \to_1^\calT \beta$ means that $\alpha$ can grow into $\beta$ by the addition of a single tile; the fact that we require both $\alpha$ and $\beta$ to be $\tau$-stable implies in particular that the new tile is able to bind to $\alpha$ with strength at least $\tau$. It is easy to check that had we instead required only $\alpha$ to be $\tau$-stable, and required that the cut of $\beta$ separating $\alpha$ from the new tile has strength at least $\tau$, then this implies that $\beta$ is also $\tau$-stable.}
If $\alpha \to_1^{\calT} \beta$, $ \dom \beta \setminus \dom \alpha=\{p\}$, and $t=\beta(p)$, we write $\beta = \alpha + (p \mapsto t)$.
The \emph{$\calT$-frontier} of $\alpha$ is the set $\partial^\calT \alpha = \bigcup_{\alpha \to_1^\calT \beta} \dom \beta \setminus \dom \alpha$, the set of empty locations at which a tile could stably attach to $\alpha$.

A sequence of $k\in\Z^+ \cup \{\infty\}$ assemblies $\vec{\alpha} = (\alpha_0,\alpha_1,\ldots)$ is a \emph{$\calT$-assembly sequence} if, for all $1 \leq i < k$, $\alpha_{i-1} \to_1^\calT \alpha_{i}$.
We write $\alpha \to^\calT \beta$, and we say $\alpha$ \emph{$\calT$-produces} $\beta$ (in 0 or more steps) if there is a $\calT$-assembly sequence $\vec{\alpha}=(\alpha_0,\alpha_1,\ldots)$ of length $k = |\dom \beta \setminus \dom \alpha| + 1$ such that
1) $\alpha = \alpha_0$,
2) $\dom \beta = \bigcup_{0 \leq i < k} \dom {\alpha_i}$, and
3) for all $0 \leq i < k$, $\alpha_{i} \sqsubseteq \beta$.
In this case, we say that $\beta$ is the \emph{result} of $\vec{\alpha}$, written $\beta=\res{\vec{\alpha}}$.
If $k$ is finite then it is routine to verify that $\res{\vec{\alpha}} = \alpha_{k-1}$.\footnote{If we had defined the relation $\to^\calT$ based on only finite assembly sequences, then $\to^\calT$ would be simply the reflexive, transitive closure $(\to_1^\calT)^*$ of $\to_1^\calT$. But this would mean that no infinite assembly could be produced from a finite assembly, even though there is a well-defined, unique ``limit assembly'' of every infinite assembly sequence.}
We say $\alpha$ is \emph{$\calT$-producible} if $\sigma \to^\calT \alpha$, and we write $\prodasm{\calT}$ to denote the set of $\calT$-producible canonical assemblies.
The relation $\to^\calT$ is a partial order on $\prodasm{\calT}$ \cite{Roth01,jSSADST}.\footnote{In fact it is a partial order on the set of $\tau$-stable assemblies, including even those that are not $\calT$-producible.}
A $\calT$-assembly sequence $\alpha_0,\alpha_1,\ldots$ is \emph{fair} if, for all $i$ and all $p\in\partial^\calT\alpha_i$, there exists $j$ such that $\alpha_j(p)$ is defined; i.e., no frontier location is ``starved''.

An assembly $\alpha$ is \emph{$\calT$-terminal} if $\alpha$ is $\tau$-stable and $\partial^\calT \alpha=\emptyset$.
It is easy to check that an assembly sequence $\vec{\alpha}$ is fair if and only $\res{\vec{\alpha}}$ is terminal.
We write $\termasm{\calT} \subseteq \prodasm{\calT}$ to denote the set of $\calT$-producible, $\calT$-terminal canonical assemblies.

A TAS $\calT$ is \emph{directed (a.k.a., deterministic, confluent)} if the poset $(\prodasm{\calT}, \to^\calT)$ is directed; i.e., if for each $\alpha,\beta \in \prodasm{\calT}$, there exists $\gamma\in\prodasm{\calT}$ such that $\alpha \to^\calT \gamma$ and $\beta \to^\calT \gamma$.\footnote{The following two convenient characterizations of ``directed'' are routine to verify.
$\calT$ is directed if and only if $|\termasm{\calT}| = 1$.
$\calT$ is \emph{not} directed if and only if there exist $\alpha,\beta\in\prodasm{\calT}$ and $p \in \dom \alpha \cap \dom \beta$ such that $\alpha(p) \neq \beta(p)$.}
We say that a TAS $\calT$ \emph{strictly self-assembles} a shape $S \subseteq \Z^2$ if, for all $\alpha \in \termasm{\calT}$, $\dom \alpha = S$; i.e., if every terminal assembly produced by $\calT$ has shape $S$.
If $\calT$ strictly self-assembles some shape $S$, we say that $\calT$ is \emph{strict}.
Note that the implication ``$\calT$ is directed $\implies$ $\calT$ is strict'' holds, but the converse does not hold.
We say that $\calT$ \emph{uniquely self-assembles} a shape $S$ if $\calT$ is directed and it strictly self-assembles $S$.

When $\calT$ is clear from context, we may omit $\calT$ from the notation above and instead write
$\to_1$,
$\to$,
$\partial \alpha$,
\emph{frontier},
\emph{assembly sequence},
\emph{produces},
\emph{producible}, and
\emph{terminal}.
We also assume without loss of generality that every single glue or double glue occurring in some tile type in some direction also occurs in some tile type in the opposite direction, i.e., there are no ``effectively null'' single or double glues.

%
%
%
%
%
%
%
%
%
%
%
%



\end{document}